\documentclass[12pt]{article}
\usepackage{geometry}                		
\usepackage[parfill]{parskip}    		
\usepackage{amssymb}
\usepackage{amsmath}

\usepackage{palatino}
\usepackage{textcomp}
\usepackage{graphicx}
\usepackage{epsfig}
\usepackage{epstopdf}
\usepackage{mathrsfs}
\usepackage[english]{babel}
\usepackage{amscd}
\usepackage{amsthm}

\usepackage{fancyhdr} 
\newtheorem{thm}{Theorem}[section]
\newtheorem{lemma}[thm]{Lemma}
\newtheorem{prop}[thm]{Proposition}

\newtheorem{conj}[thm]{Conjecture}

\newtheorem{prbl}[thm]{Problem}

\title{A generalized isodiametric problem}

\author{Christos Pelekis\thanks{Department of Computer Science, KU Leuven
Celestijnenlaan 200A,
3001, Belgium, Email: pelekis.chr@gmail.com}}

\begin{document}
\maketitle

\abstract{Fix positive integers $a$ and $b$ such that $a> b\geq 2$ and a positive real $\delta>0$. 
Let $S$ be a planar set of diameter $\delta$ having the  
following property: for every $a$ points in $S$, at least $b$ of them have pairwise distances 
that are all less than or equal to $2$.
What is the maximum Lebesgue measure of $S$? In this paper we investigate this problem.
We discuss the, devious, motivation that leads to its formulation and provide  upper bounds
on the Lebesgue measure of $S$. Our main result is based on a generalisation of a
theorem that is due to Heinrich Jung. 
In certain instances we are able to find the extremal set but the general 
case seems elusive.}\\

\section{Motivation, problem statement and main result}

The main purpose of this article is to draw the attention of the reader to a 
particular extremal problem in geometry.
The setting will be the Euclidean plane.  Plane geometry is abundant in extremal problems. 
Such problems have been studied extensively throughout history and have been generalised in many ways.
An instance of an extremal problem in plane geometry is the so-called \emph{isodiametric problem}: 
Fix a positive real, $\Delta$. 
Among all measurable sets in the plane whose
diameter equals $\Delta$, find one that has maximum Lebesgue measure. 
This is a well known problem whose solution implies that disks maximise Lebesgue measure 
under restriction on the diameter.
We refer the reader to the excellent book of Peter 
Gruber \cite{Gruber} for details and 
references.
In this paper we address a generalisation of the isodiametric problem. Let us begin
with the motivating question that leads to its formulation and, on the same time, 
introduce some notation 
that will be fixed throughout the text. 
Here and later, we denote by $D(p,r)$  the closed disk in the plane with centre at the point 
$p\in\mathbb{R}^2$ and radius $r>0$, i.e., 
$D(p,r)$ consists of all points in the plane that are at distance $\leq r$ from $p$.  

Suppose that you want to poison your mother-in-law.
You baked a pie that has the shape of a closed disk of radius $R>2$; let the disk
$D(O,R)$ denote the pie, where $O$ is the origin of the plane. 
Your mother-in-law is going to eat a circular piece,
strictly contained in the pie, of radius $1$ which is chosen uniformly at random. 
More precisely, she chooses uniformly at random a point $p\in D(O,R-1)$ 
and eats the set $D(p,1)$.
You possess $h\geq 1$ grams of arsenic whose lethal dose is, say, $1$ gram. 
You can distribute the poison in any way you want over the pie. Which distribution of poison has the 
highest probability of doing the old lady?

Let us refer to the previous puzzle as the \emph{poisoning problem}.
It turns out that the poisoning problem gives rise to a generalisation of the isodiametric problem. 
In order to illustrate this,
suppose for a moment that the available amount of poison satisfies $1\leq h < 2$. 
We will refer to a disk of radius $1$ as a \emph{unit disk}  and
to a unit disk that contains at least $1$ gram of poison as a  \emph{lethal} unit disk. 
If $1\leq h < 2$ then for every
distribution of poison over the pie any two lethal unit disks, $U_1,U_2$, must
have non-empty intersection; this is equivalent to saying that 
their centres are at distance $\leq 2$. Now the poisoner has to distribute the arsenic 
in a way that maximises the Lebesgue measure of the set consisting of the centres of lethal disks. 
Additionally, the centres of those disks must have pairwise distances that are less than or equal to $2$.
What is the maximum Lebesgue measure of a set of points in the plane
any two of which are at distance $\leq 2$?
Now, any two points of a set are at distance $\leq 2$ if and only if
the diameter of the set is $\leq 2$. Hence, when $1\leq h<2$, the poisoning problem is
equivalent to the following problem whose solution implies that disks maximise
area under constrains on the diameter. 

\begin{prbl}[Isodiametric problem] Among all planar sets whose diameter is $\leq 2$, find one that
has maximum Lebesgue measure. 
\end{prbl}

The solution is given by a disk of diameter $2$. Hence the Lebesgue measure of a set of diameter 
$\leq 2$ is at most $\pi$; the latter bound is referred to as the \emph{isodiametric inequality} (see \cite{Gruber}). 
Notice that the isodiametric inequality implies that the 
poisoner maximises the probability of doing the mother-in-law  by
putting $1$ gram of arsenic at, say, the centre $O$ of the pie.
More generally, let $a$ be a non-negative integer and suppose that the amount of 
poison, $h$, that is available to the poisoner satisfies 
$a-1\leq h<a$. Then, for any distribution of poison over the pie,
one cannot find $a$ pairwise disjoint lethal unit disks
and so  the centres of the lethal disks form a set having  the property that among any $a$ of its points, 
at least two are at distance $\leq 2$.  The phrase "at least two" in the last sentence could be replaced by 
"at least $b$" and so
we arrive at the formulation of the following.

\begin{prbl}[A generalised isodiametric problem]
\label{probone}
Fix $\delta>0$ and positive integers $a>b$. Assume that $S$ is a planar set of 
diameter $\delta$ having 
the property that
among any $a$ points in $S$, at least $b$ of them are such that 
their pairwise distances are  all
less than or equal to $2$. What is the
maximum Lebesgue measure of $S$? 
\end{prbl}

The constraint on the diameter of $S$ is inspired by the assumption that the pie has radius $R$ 
and so its diameter equals $\delta= 2R$. 
From now on we will only focus on Problem \ref{probone}. 
Let us introduce some extra notation.  
Given positive integers $a>b$, we will say that a subset, $S$, 
of the plane is a $T(a,b)$-\emph{set} if 
among any $a$ points in $S$, at least $b$ of them 
are such that their pairwise distances are all less than or equal to $2$.
In this paper we mainly focus on $T(3,2)$-sets. 
For fixed $\delta>0$ and positive integer $a,b$ such that $a>b\geq 2$, 
we will say that a set $S\subseteq \mathbb{R}^2$ is $T(a,b)$-\emph{extremal} of diameter $\delta$ if it 
is a $T(a,b)$-set of maximal Lebesgue measure 
whose diameter is equal to $\delta$. 
Finally, we will denote by $\lambda_2(S)$ the two-dimensional Lebesgue measure of a set $S\subseteq \mathbb{R}^2$.
The main results of this article are summarised in the following. 

\begin{thm}\label{main} Fix $\delta >0$. Assume that $S\subseteq \mathbb{R}^2$ is a $T(3,2)$-set 
whose diameter equals $\delta$. 
The following hold true.
\begin{enumerate}
\item If $\delta \leq \frac{4}{\sqrt{3}}$, then $\lambda_2(S) \leq \frac{\pi \delta^2}{4}$. A $T(3,2)$-extremal set 
of diameter $\delta$ is 
a disk of diameter $\delta$.
\item  If $\delta \geq 4$, then $\lambda_2(S) \leq 2\pi$. A $T(3,2)$-extremal set of diameter $\delta$
 is the union of two disjoint unit disks. 
\item  If $\frac{4\pi}{3}< \delta < 4$, then $\lambda_2(S) \leq 
\min\left\{\frac{\pi}{6}\frac{\delta^4}{\delta^2-1}+\frac{4\pi}{9}\;,\; 2\pi \right\} $.
\item If  $S$ is \emph{convex}, then 
$\lambda_2(S)\leq  \min\left\{\frac{4\pi}{3\sqrt{3}}\cdot\delta, \; 
2\pi \right\}$, for $\delta>\frac{4}{\sqrt{3}}$.
\item If $S$ is \emph{symmetric}, then $\lambda_2(S)\leq \min\left\{\frac{\pi \delta^2}{6} + \frac{4\pi}{9}, 2\pi \right\}$, for $\delta>\frac{4}{\sqrt{3}}$.
\end{enumerate}
\end{thm} 

We break the proofs of the above statements into several sub-theorems that are dispersed  throughout the paper. 
In particular, in Section \ref{general} we  provide proofs of the first three statements. 
The proofs of the first two statements are quite simple and the proof of the third one is 
based on a generalisation of a theorem that is due to Heinrich Jung (see \cite{Jung}). 
In Section \ref{convex} we prove the fourth statement which is obtained as a direct application 
of a theorem that is due to Wilhelm Blaschke. In Section \ref{symmetric} we recall the definition 
of a symmetric set and  prove the final statement, which is
an application of polar coordinates. 
In Section \ref{abproperty}  we find
$T(a,2)$-extremal sets, where $a\geq 3$, and we use this result to obtain sets that are
$T(N_b,b)$-extremal, where $N_b= (b-1)a-b+2$, for $a\geq 3$, and $b\geq 2$.
Finally, in Section \ref{conjecture}, we state some 
open problems.

\section{Measurable sets}\label{general}

In this section we prove the first three statements of Theorem \ref{main}. 
The proof of the third statement will be based on a generalisation of theorem that is due to Heinrich Jung. 
We will also formulate and employ a 
corresponding problem on the circle. 
Let us begin with the following, well known, result that will be used in the proof of the first statement.

\begin{thm}[Brunn-Minkowski inequality] Let $S_1$ and $S_2$ be non-empty measurable sets in $\mathbb{R}^2$. 
Then
\[ \lambda_2(S_1+S_2)^{1/2} \geq \lambda_2(S_1)^{1/2} + \lambda_2(S_2)^{1/2} , \]
where $S_1+S_2=\{s_1+s_2: s_1\in S_1,s_2\in S_2\}$.
\end{thm}
\begin{proof} See \cite{Gruber} or \cite{Matousek}.
\end{proof}

The first statement of Theorem \ref{main} is obtained as an 
application of the Brunn-Minkowski inequality. 

\begin{thm} Let $S$ be a planar $T(3,2)$-set of diameter $\delta$. 
Suppose that $0<\delta \leq \frac{4}{\sqrt{3}}$. 
Then $\lambda_2(S)\leq \pi \left(\frac{\delta}{2}\right)^2$. The bound is attained by
the measure of a disk of diameter $\delta$. 
\end{thm}
\begin{proof} If $0<\delta\leq 2$, then the result follows from the isodiametric problem.
So suppose that $2<\delta\leq \frac{4}{\sqrt{3}}$. 
It is not difficult to see that the smallest disk that contains
an equilateral triangle of side length $2$ is a disk of diamter $\frac{4}{\sqrt{3}}$. 
Clearly, such a disk is a $T(3,2)$-set and so the same holds for every disk of diameter
$\delta \leq \frac{4}{\sqrt{3}}$.
If $S$ is set of diameter $\delta$ then Brunn-Minkowski inequality implies 
\[ 4\lambda_2(S) \leq \lambda_2(S-S) , \]
where $S-S:=\{s_1-s_2: s_1,s_2\in S\}$. Now the hypothesis that
the diameter of $S$ is $\leq \delta$ implies that $S-S\subseteq D(0,\delta)$ and so
\[  4\lambda_2(S) \leq \lambda_2(S-S) \leq \pi \delta^2 \; \Rightarrow \; \lambda_2(S)\leq \pi \left(\frac{\delta}{2}\right)^2 . \]
Since the latter bound is attainable by a disk of diameter $\delta$, the result follows.
\end{proof}

We now proceed with the proof of the second statement. 

\begin{thm}
\label{mainprop} Fix $\delta\geq 4$. Suppose that $S\subseteq \mathbb{R}^2$ is a $T(3,2)$-set with the
whose diameter equals $\delta$.
Then  $\lambda_2(S)\leq 2\pi$. A $T(3,2)$-extremal set of diameter $\delta$ is
a union of two disjoint unit disks. 
\end{thm}
\begin{proof} Choose two points $p_1,p_2\in S$, which are at distance $\delta$.
Since $\delta\geq 4$, it follows that the two disks $D(p_1,2)$ and $D(p_2,2)$ intersect in at most one point.
Now notice that no point of $S$ belongs outside the set $D(p_1,2)\cup D(p_2,2)$ since such a point 
forms with the points $p_1$ and $p_2$ a triangle whose sides have all length $>2$. Hence all points
of $S$ belong  either to $D(p_1,2)$ or to $D(p_1,2)$. 
Notice also that no two points of $S\cap D(p_1,2)$ (resp. of $S\cap D(p_2,2)$)
can be at distance $>2$, because they would form with $p_2$ (resp. with $p_1$)
a triangle whose sides have all length $>2$. Thus any two points in $S \cap D(p_1,2)$ and any two points in 
$S \cap D(p_2,2)$
are at distance $\leq 2$ and so the isodiametric inequality yields
$\lambda_2(S\cap D(p_1,2)),\lambda_2(S\cap D(p_2,2))\leq \pi$.
Since  $D(p_1,2)\cap D(p_2,2)=\emptyset$, we conclude that 
$\lambda_2(S)=\lambda_2(S\cap D(p_1,2))+\lambda_2(S \cap D(p_2,2))\leq 2\pi$, as required.
\end{proof}

We have thus found a $T(3,2)$-extremal set of diameter $\delta \in (0,4/\sqrt{3}]\cup [4,\infty)$. 
We were \emph{not} able find a $T(3,2)$-extremal set of diameter $\delta\in (4/\sqrt{3},4)$; we conjecture 
its shape in Section \ref{conjecture}. This leads us in search for upper bounds on 
the Lebesgue measure of a $T(3,2)$-extremal set of diameter $\delta\in (4/\sqrt{3},4)$. 
This is the content of the third statement of Theorem \ref{main}.
In order to prove the third statement we will first prove a generalisation of a result that is due to Heinrich Jung. 
Let us first introduce some extra notation. Given $S\subseteq \mathbb{R}^2$ we will denote by 
$\text{diam}(S)$ the 
diameter of $S$, i.e., the supremum of the distances between two points of $S$. Formally, 
\[ \text{diam}(S) = \sup_{p,q\in S} \; |p-q| , \] 
where $|p-q|$ denotes the distance between the points $p$ and $q$. 
The upcoming result of Jung answers the following question: Suppose that 
$S\subseteq \mathbb{R}^2$ is 
a bounded set. What is the smallest radius of a disk that contains $S$? 

\begin{thm}[Jung, 1901] Let $S\subseteq \mathbb{R}^2$ be bounded. If $\text{diam}(S) =\delta$,
then $S$ is contained in a disk of radius $\frac{\delta}{\sqrt{3}}$.
\end{thm}
\begin{proof} See Jung \cite{Jung} for the earliest proof of this result, Gruber \cite{Gruber} for the 
high-dimensional 
analogue and  Szil\'ard \cite{Szilard} for an accessible proof of the planar case.
\end{proof}

We remark that Jung's theorem is sharp in the sense that an equilateral triangle of side length $\delta$ 
attains the bound.
We will prove, in a subsequent statement, a generalisation of Jung's result. Before doing so, let us 
recall  the following Helly-type result that is due to Victor Klee.

\begin{thm}[Klee]
\label{Klee} Let $\mathcal{F}$ be a family of compact convex subsets of $\mathbb{R}^2$ containing at 
least $3$ members. Suppose that $K$ is a compact convex subset of $\mathbb{R}^2$ such that 
for each subfamily of $3$ sets in $\mathcal{F}$ there exists a translate of $K$ that contains all $3$ of them.
Then there exists a translate of $K$ that contains all members of $\mathcal{F}$.
\end{thm}
\begin{proof} See \cite{Lay}, Theorem $6.5$.
\end{proof}

The proof of the third statement of Theorem \ref{main} will be based upon the following 
generalisation of the notion of diameter of a planar set.
Define the $3$-\emph{diameter} of a set $S\subseteq \mathbb{R}^2$ by setting
\[ \text{diam}_3(S) = \sup_{\{p_1,p_2,p_3\}\subseteq S} \; 
\min\left\{ |p_1-p_2|,  |p_2-p_3|, |p_1-p_3|\right\} . \]
Note that a set $S\subseteq \mathbb{R}^2$
is a  $T(3,2)$-set if and only if $\text{diam}_3(S) \leq 2$. Notice also that 
$\text{diam}_3(S) \leq \text{diam}(S)$, for any $S\subseteq \mathbb{R}^2$.
The next result is an extension of Jung's theorem  
that takes into account additional information on the $3$-diameter of the set.

\begin{figure}[htbp]
   \centering
   \includegraphics[width=0.80\textwidth]{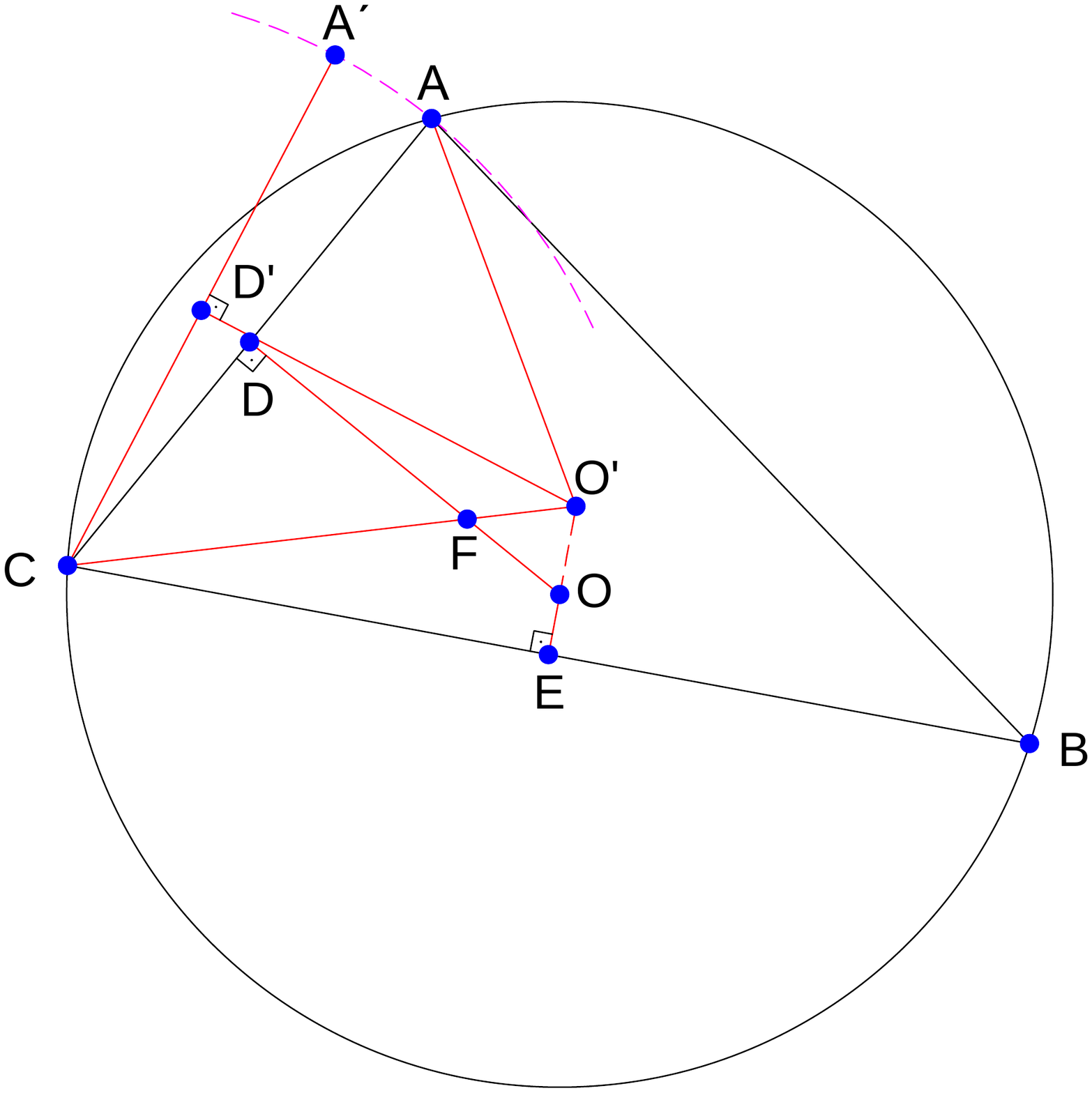} 
   \caption{\small{Points $A,B,C$ in a set, $S$, with $\text{diam}(S)=\delta$ and $\text{diam}_3(S)=\tau$}}
   \label{fig:example}
\end{figure}

\begin{thm}\label{jung} Suppose that $S\subseteq \mathbb{R}^2$ is bounded. 
If $\text{diam}(S) =\delta$ and $\text{diam}_3(S)=\tau \leq \delta$,
then $S$ is contained in a disk of radius 
$\rho:=\frac{1}{2}\frac{\delta^2}{\sqrt{\delta^2 - \frac{\tau^2}{4}}}$.
\end{thm}
\begin{proof} If $\tau = \delta$ then the result reduces to Jung's; so we may assume $\tau < \delta$.
It is not difficult to see that $\rho$ is the radius of the disk that 
circumscribes an isosceles triangle whose equal sides have length $\delta$ and the other side 
has length $\tau$; therefore the result is sharp. Theorem \ref{Klee} implies that it is enough to
show that given any three points  $A,B,C$ of $S$, there exists a disks of 
radius $\rho$ that covers the points $A,B,C$. Fix three points $A,B,C$ from $S$.
If the triangle $ABC$ is obtuse angled, right angled or 
degenerate, then it is easily seen that 
the circle whose diameter is the largest side of $ABC$ covers the triangle and 
the radius of such a circle is easily seen to be $\frac{\delta}{2}< \rho$. So we may assume that the 
triangle $ABC$ is acute angled;
see Figure \ref{fig:example}. Assume, without loss of generality, that the angle $\hat{A}$ is maximal 
and that the length, $t$, of the side $AC$ is minimal; hence the assumption $\text{diam}_3(S)=\tau$
implies that $t\leq \tau$.
Now let $\ell_1$ be the length of the segment $AB$, $\ell_2$ 
the length of the segment $BC$ and notice that $t\leq \ell_1\leq \ell_2\leq \delta$. 
We may also assume that $\ell_2 > \tau$, since if this is not the case, then  
the points $A,B,C$ 
are contained in a disk of radius 
$\frac{\tau}{\sqrt{3}}=\frac{1}{2}\frac{\tau^2}{\sqrt{\tau^2 - \tau^2/4}}$;
since  the function $f(x)= \frac{x^2}{\sqrt{x^{2} -\tau^2/4} }$ is 
increasing when $x\geq \tau$ it follows that the later quantity is the less than $\rho$.
Let $O$ be the centre of the circle that circumscribes the triangle $ABC$
and $r$ its radius.
Let also $OE$ be the perpendicular bisector of the side $BC$. 
There are two cases to consider: either $\ell_1 \geq \tau$ or $\ell_1 < \tau$. \\
(a) Suppose first that $\ell_1 \geq \tau$.
If $\ell_1 =\ell_2$ then it is easy to see that 
 $r = \frac{\ell_1^2}{2\sqrt{\ell_{1}^{2} -t^2/4}} \leq  \frac{\ell_1^2}{2\sqrt{\ell_{1}^{2} -\tau^2/4}}$ and 
the result follows from the fact that the function $f(x)$ is 
increasing when $x\geq \tau$. So suppose that $\ell_1 < \ell_2$. 
In that case we claim that the points $A,B,C$ are 
contained in a disk of radius  $\frac{\ell_2^2}{2\sqrt{\ell_{2}^{2} -t^2/4} }$ 
and therefore the result follows, as before, from the fact that the function $f(x)$ is increasing when $x\geq \tau$. 
To prove the claim, let $A^{\prime}$ be the point, obtained by rotating the segment $CA$ around $C$
in the counterclockwise direction, for which  
the length of the segment $A^{\prime}B$ is equal to the length of the segment 
$BC$, i.e. equal to $\ell_2$ (see Figure \ref{fig:example}). Since 
$\ell_1 <\ell_2$ it follows that the points $A$ and $A^{\prime}$ are different. 
Let $O^{\prime}$ be the centre of the circle, $\Sigma$, 
that circumscribes the triangle $A^{\prime}BC$ and 
let $D^{\prime}O^{\prime}$ 
be the perpendicular bisector of the side $A^{\prime}C$.  
The proof of the claim will follow once we show that the 
point $A$ belongs to the interior of the circle $\Sigma$. To see this, let $F$ be the point at which the 
segments $CO^{\prime}$ and $DO$ intersect and look at the triangle $AFO^{\prime}$. The triangle 
inequality implies $\lambda_1(AO^{\prime}) \leq \lambda_1(AF) + \lambda_1(O^{\prime}F)$, where 
$\lambda_1(\cdot)$ denotes the length of the segment. Since $F$ belongs to the perpendicular bisector 
of the segment $AC$ it follows $\lambda_1(AF) = \lambda_1(CF)$ and therefore 
\[\lambda_1(AO^{\prime}) \leq \lambda_1(CF) + \lambda_1(O^{\prime}F) = \lambda_1(CO^{\prime}). \]
This implies that the point $A$ belongs to the interior of the circle $\Sigma$ and the claim follows. \\
(b) If $\ell_1 < \tau$ then the argument is similar and so we briefly 
sketch it. Recall that the triangle is assumed to be acute angled 
and notice that in this case we have $t<\tau, \; \tau<\ell_2<2\tau$. 
Notice that the point $A$ is contained in the intersection of the two disks $D(C,\tau)$ and $D(B,\tau)$ and 
that the radius of the circle that circumscribes $ABC$ is maximum when the point $A$ lies on the 
boundary of $D(C,\tau)\cap D(B,\tau)$. 
Suppose that the line that passes through $B$ and $C$ divides the plane in a positive and negative part 
and assume, 
without loss of generality, that 
$A$ lies on the positive part. Let $A^{\prime}$ be the point in the positive part at which the 
boundaries of disks $D(C,\tau)$ and $D(B,\tau)$ intersect and assume that $A\neq A^{\prime}$. 
Let us also assume that the distance between $B$ and $A$ is smaller than 
the distance between $C$ and $A$. 
Now using a similar argument as in (a) one can    
see that the triangle $ABC$ is contained in the circle that circumscribes the  
isosceles triangle $A^{\prime}BC$.
Repeating the argument of part (a) again, i.e. rotate the segment   
$CA^{\prime}$ around $C$
in the counterclockwise direction until you hit a point $A^{\prime \prime}$ such that the length 
of $A^{\prime\prime}B$ equals $\ell_2$, shows that the triangle $A^{\prime}BC$ is contained 
in the disk that circumscribes the triangle $A^{\prime\prime}BC$ 
whose radius equals $\frac{\ell_2^2}{2\sqrt{\ell_{2}^{2} -\tau^2/4} }$.
The result follows. 
\end{proof}

In order to prove the third statement of Theorem \ref{main}, we formulate 
a corresponding isodiametric problem on the circle.

\begin{prbl}[The problem on the circle]
\label{circle}
Fix a real number $r>\frac{2}{\sqrt{3}}$ and let $C(r)$ denote the perimeter of a circle of radius $r$.
Suppose that $B\subseteq C(r)$  is a $T(3,2)$-set. What is the maximum 
($1$-dimensional) Lebesgue measure of $B$?  
\end{prbl}

The next result gives an upper  bound on 
the measure of $B$.

\begin{lemma}
\label{cycle} Let $B$ and $r$ be as in Problem \ref{circle}. Then 
\[ \lambda_1(B) \leq \frac{2}{3}\cdot 2\pi r , \]
where $\lambda_1(\cdot)$ denotes $1$-dimensional Lebesgue measure.
\end{lemma} 
\begin{proof} Fix three points $p_1,p_2,p_3 \in C(r)$ such that the triangle they form is 
an equilateral triangle, $T$, of maximum side length. Since $r>\frac{2}{\sqrt{3}}$, 
all sides of $T$ have length $>2$ and so at most two 
points from the set $\{p_1,p_2,p_3\}$ can belong to $B$. The same 
is true for the three vertices of any triangle, $T_{\varphi}$, obtained 
by rotating the triangle 
$T$ by $\varphi \in (0,120)$ degrees and the vertices of this family of triangles cover the set $C(r)$. 
This implies that the set $B$ can have at most $\frac{2}{3}$ of the total measure of $C(r)$, as required.
\end{proof}

We now have all the necessary tools to prove the third statement of Theorem \ref{main}.

\begin{thm}\label{polar} Suppose that $S\subseteq \mathbb{R}^2$ is a $T(3,2)$-set.
Let $\text{diam}(S)=\delta \in \left(\frac{4\pi}{3}, 4\right)$. Then
\[ \lambda_2(S) \leq \min\left\{\frac{\pi}{6}\frac{\delta^4}{\delta^2-1}+\frac{4\pi}{9}\;,\; 2\pi \right\} . \]
\end{thm}
\begin{proof}  Since $S$ is a  $T(3,2)$-set
it follows that $\text{diam}_3(S) \leq2$.
From Theorem \ref{jung} we know that $S$ is contained 
in a ball of radius $\rho := \frac{1}{2}\frac{\delta^2}{\sqrt{\delta^2-1}}$. Using polar coordinates we write
\[ \lambda_2(S)= \int_{0}^{\rho} \lambda_1(S\cap C(r)) \;dr = 
\int_{0}^{2/\sqrt{3}}+ \int_{2/\sqrt{3}}^{\rho} \lambda_1(S\cap C(r)) \;dr.\]
The first integral is at most 
\[\int_{0}^{2/\sqrt{3}}2\pi r\; dr =\frac{4\pi}{3} . \]
Now $S\cap C(r)$ is clearly a $T(3,2)$-set. Therefore,
using Lemma \ref{cycle}, we conclude that the second integral is
\[ \leq  \int_{2/\sqrt{3}}^{\rho} \;\frac{2}{3}\cdot 2\pi r \;dr =
 \frac{\pi}{6}\frac{\delta^4}{\delta^2-1}-\frac{8\pi}{9}. \] 
Summarising, we have shown that
\[\lambda_2(S) \leq \frac{\pi}{6}\frac{\delta^4}{\delta^2-1}+\frac{4\pi}{9}. \]
The result follows from Theorem \ref{mainprop} and the fact that the maximal measure of 
a $T(3,2)$-set  is an increasing function of the diameter.
\end{proof}

One can easily see that $\frac{\pi}{6}\frac{\delta^4}{\delta^2-1}+\frac{4\pi}{9}\leq 2\pi$, for $\delta\leq 2.85$ 
and so the previous bound is better than the bound of Theorem \ref{mainprop}.
In the next section we consider Problem \ref{probone} under the additional assumption that 
the set $S$ is convex. 

\section{Convex sets}\label{convex}

In this section we consider Problem \ref{probone} under the additional assumption that 
$S$ is a convex set of diameter $\delta\in (\frac{4}{\sqrt{3}}, 4)$. 
In particular, 
we prove the fourth statement of Theorem \ref{main}. 
We begin with some definitions. 
Suppose that $S\subseteq \mathbb{R}^2$ is a \emph{convex} $T(3,2)$-set   
such that $\text{diam}(S)=\delta$.
For any three points $p_1,p_2,p_3\in S$ let us denote by $\text{tr}(p_1,p_2,p_3)$ the 
area of the triangle formed by $p_1,p_2,p_3$. 
Define the \emph{triameter} of $S$, denoted $\text{tr}(S)$, to be the quantity
\[ \text{tr}(S) = \sup_{\{p_1,p_2,p_3\}\subseteq S} \text{tr}(p_1,p_2,p_3) . \]
Now fix three points in $S$, say $p_1,p_2,p_3$.
By assumption we know at least two points from 
the set $\{p_1,p_2,p_3\}$ have distance $\leq 2$. Suppose these points are $p_1$ and $p_2$. 
Then the distance of the point $p_3$ from the line segment joining $p_1$ and $p_2$ is at most $\delta$.
Hence 
\[ \text{tr}(p_1,p_2,p_3) \leq \frac{1}{2} \cdot 2\cdot\delta = \delta . \]
We now recall the following result of Wilhelm Blaschke.

\begin{thm}[Blaschke, 1923]
\label{Blaschke} Let $S$ be a convex subset of the plane such that $\text{tr}(A)=\delta$.
Then 
\[ \lambda_2(S) \leq \frac{2\pi}{3} \cdot \frac{1}{\sin(2\pi/3)} \cdot \delta = \frac{4\pi}{3\sqrt{3}}\delta. \]
\end{thm}
\begin{proof} See \cite{Chakerian} or \cite{Pach}, Theorem $2.6$.
\end{proof}

Notice that for $\delta \leq \frac{3\sqrt{3}}{2}\approx 2.59$ the bound of the previous Lemma is $\leq 2\pi$. 
Blaschke's result settles the fourth statement of Theorem \ref{main}.
In fact, as is shown in the following result, one can prove a 
slightly better bound, by combining Blaschke's theorem with Theorem \ref{jung}.

\begin{prop}
If $S\subseteq \mathbb{R}^2$ is a \emph{convex}  $T(3,2)$-set of diameter $\delta$, then 
\[ \lambda_2(S) \leq \min\left\{\frac{\pi}{4}\cdot \frac{\delta^4}{\delta^2 -1}, \; 2\pi \right\} .\]
\end{prop}
\begin{proof} We may assume that $\delta > 4/\sqrt{3}$, since otherwise the first statement of Theorem \ref{main} applies. 
Theorem \ref{jung} implies that $S$ is contained in a disk of radius 
$\rho :=\frac{1}{2}\frac{\delta^2}{\sqrt{\delta^2 -1}}$. Now it is well-known, and easy to see, that a triangle 
of maximal area that is contained in a disk of radius $\rho$ is an equilateral one whose sides are 
equal to $\rho \sqrt{3}$. Therefore, the maximum area of a triangle that is contained in $S$ is at most 
$\frac{3\sqrt{3}}{16}\frac{\delta^4}{\delta^2 -1}$ and the result follows from Theorem \ref{Blaschke}. 
\end{proof}

Notice that for $\delta \leq 2.612$ we have $\frac{\pi}{4}\cdot \frac{\delta^4}{\delta^2 -1}\leq  2\pi$.
Let us remark that we were not able to find a $T(3,2)$-extremal and convex set of 
diameter $\delta\in (\frac{4}{\sqrt{3}}, 4)$. We conjecture that such an extremal set is given 
by the convex hull of the set formed by a segment of length $\delta$ and a unit disk centred 
at its midpoint.  
In the next section we consider Problem \ref{probone} under the additional assumption that 
the set $S$ is symmetric.

\section{Symmetric sets}\label{symmetric}

We call a subset, $S$, of the plane \emph{symmetric} (with respect to the origin) 
if it has the property that whenever $x\in S$, then also $-x\in S$. Note that the union of two 
disjoint unit disks as well as the union of two intersecting unit disks is a symmetric set.

\begin{thm} Let $S\subseteq \mathbb{R}^2$ be a symmetric $T(3,2)$-set.
Suppose the diameter of $S$ is $\delta > \frac{4}{\sqrt{3}}$. Then 
\[ \lambda_2(S) \leq \min\left\{\frac{\pi \delta^2}{6} + \frac{4\pi}{9}, 2\pi \right\} . \]
\end{thm}
\begin{proof} 
Since $S$ is symmetric it follows that it is contained in a disk of radius $\leq \delta/2$. 
Now using polar coordinates we can write
\[ \lambda_2(S) = \int_{0}^{\delta/2} \lambda_1(S\cap C(r))\; dr = \int_{0}^{\frac{2}{\sqrt{3}}} + \int_{\frac{2}{\sqrt{3}}}^{\delta/2}  \lambda_1(S\cap C(r))\; dr . \]
and the argument proceeds along the same lines as in Theorem \ref{polar}.
Note that for $\delta \leq \sqrt{28/3}\approx 3.055$ we have  
$\frac{\pi \delta^2}{6} + \frac{4\pi}{9} \leq 2\pi$. 
\end{proof}

Hence the fifth statement of Theorem \ref{main} follows. In the next section we state some 
results regarding
$T(a,b)$-extremal sets, for particular values of $a,b$.

\section{$T(a,2)$-extremal sets}
\label{abproperty}

We begin this section with the following version of Theorem \ref{mainprop}.

\begin{thm}\label{atwo}  Fix $a\geq 3$. Suppose that $S\subseteq \mathbb{R}^2$ is a $T(a,2)$-set.
Then we have $\lambda_2(S) \leq (a-1) \pi$. A $T(a,2)$-extremal set is a 
union of $a-1$ pairwise disjoint unit disks.
\end{thm}
\begin{proof} Notice that we make no assumption on the diameter of $S$. 
Among all $T(a,2)$-sets in the plane choose one 
for which 
there exist $a-1$ points, $p_1,\ldots,p_{a-1}\in S$ such that the distance between 
any two different points, $p_i, p_j$, satisfies $|p_i-p_j|> 4$;
an example of such a set is a union of $a-1$ pairwise disjoint unit disks.
We now show that this is an extremal set.
Note that the disks $D(p_i,2), i=1,\ldots,a-1$ are pairwise disjoint. 
As in Theorem \ref{mainprop}
we have $S\subseteq \cup_{i=1}^{a-1} D(p_i,2)$ since if there is any point of $S$ outside the union, then it 
would form together with the points   $p_1,\ldots,p_{a-1}$ a set of $a$ points whose pairwise distances 
are all $>2$.
Similarly, any two points in the set $S \cap D(p_i,2)$ are at distance $\leq 2$.
The isodiametric inequality finishes the proof.
\end{proof}

We remark that the previous proof works for higher-dimensional subsets, $S\subseteq \mathbb{R}^d,d\geq 2$.
Theorem \ref{atwo} can be employed in order to find $T(N_b,b)$-extremal sets, where $N_b$ is a particular 
positive integer that depends on $b$.  

\begin{thm} Fix positive integers $a,b$ such that $a\geq 3$ and $b\geq 2$. 
Set $N_b=(b-1)a-b+2$. Suppose that $S\subseteq \mathbb{R}^2$ is a $T(N_b,b)$-set.
Then $\lambda_2(S)\leq (a-1)\pi$. A $T(N_b,b)$-extremal set 
is a union of $a-1$ pairwise disjoint unit disks.
\end{thm}
\begin{proof} 
We induct on $b$.
The base case, $b=2$, is treated in Theorem \ref{atwo}. Fix $b>2$ and
assume that the result holds true for all positive integers in the set $\{2,\ldots,b-1\}$ and any $a\geq 3$. 
Let $S$ be an $T(N_b,b)$-set  whose Lebesgue 
measure is maximal. Clearly, the union of $a-1$ pairwise disjoint 
unit disks is a $T(N_b,b)$-set; hence $\lambda_2(S)\geq (a-1)\pi$ and it remains to 
show that $\lambda_2(S)\leq (a-1)\pi$. We now claim that $S$ is a $T(N_{b-1},b-1)$-set;
the result then follows from the inductional hypothesis. 
To prove the claim, assume that $S$ is not an $T(N_{b-1},b-1)$-set and so there exist 
$N_{b-1}$ points in $S$, say $A_i,i=1,\ldots,N_{b-1}$, such that among any $b-1$ of them, 
there exists a pair of points whose distance is $>2$. 
Now choose $a-1$ points, say $p_1,\ldots,p_{a-1}$, in $S$ that 
are different from the points $A_i,i=1,\ldots,N_{b-1}$. 
We claim that among the points $p_1,\ldots,p_{a-1}$  there exists a pair whose distance is $\leq 2$.
To see this, suppose that all pairs from the set  $\{p_1,\ldots,p_{a-1}\}$ are at distance $>2$. 
Then each subset of  
the set $\{A_1,\ldots,A_{N_{b-1}},p_1\ldots, p_{a-1}\}$ of cardinality $b$ contains
a pair of points that are at distance $>2$, which
contradicts  the assumption that $S$ is a $T(N_b,b)$-set. 
Since the points $p_1,\ldots,p_{a-1}$ were arbitrary, we conclude that
the set $S\setminus \{A_1,\ldots,A_{N_{b-1}}\}$ is a $T(a-1,2)$-set.  
Now Theorem \ref{atwo} yields
$\lambda_2(S) = \lambda_2(S\setminus \{A_1,\ldots,A_{N_{b-1}}\})\leq (a-2)\pi$, which 
contradicts the maximality of 
the Lebesgue measure of $S$. Hence $S$ is a $T(N_{b-1},b-1)$-set and the result follows.
\end{proof}

Our paper ends with the next section in which  
we collect some open problems.

\section{Open problems}\label{conjecture}

So far we have formulated the "poisoning problem", Problem \ref{probone} and Problem \ref{circle}. 
These problems are, to our knowledge, 
generally open.  
Moreover, they are well defined on high-dimensional spaces; one has  to
replace perimeter by surface and $\mathbb{R}^2$ by $\mathbb{R}^d$ in their formulation. 
Let us briefly mention the existence of yet another poisoning problem, which is also open, where 
the poisoner bakes biscuits instead of a pie; the interested reader is referred to Fokkink et al. \cite{Fokkink} for 
details and references.  Let us also mention that a problem which is similar to the 
$3$-dimensional version of Problem \ref{circle} 
has been considered by B\'ela Bollob\'as in \cite{Bollobas},  but the main result therein has been withdrawn 
as incorrect, (see Balogh et al. \cite{Balogh}, page $47$).

In the two-dimensional case we saw that, given $\delta>0$,
an extremal $T(3,2)$-set of diameter $\delta$ is a disk of radius $\delta/2$, for $\delta\in (0,4/\sqrt{3}]$ and,
for $\delta\geq 4$, an extremal $T(3,2)$-set of diameter $\delta$ is the union of two unit disks. 
The reader may suspect that, in case 
$\delta\in (4/\sqrt{3}, 4)$, an extremal $T(3,2)$-set of diameter $\delta$ 
is given by the union of two intersecting unit disks; 
where the disks intersect 
in such a way that their union is a set of diameter $\delta$.  
We believe that this is indeed the case.

\begin{conj} Let $U_{\delta}$ be the union of two unit disks whose intersection is non-empty. Assume that
$\text{diam}(U_{\delta})= \delta \in (4/\sqrt{3},4)$. Then $U_{\delta}$ is 
a $T(3,2)$-extreamal set of diameter $\delta$.
\end{conj}

We were not able to prove this conjecture. This led us in search for bounds on the maximal measure 
and we succeeded in doing so by employing a generalisation of Jung's theorem, namely Theorem \ref{jung}.
Perhaps the reader will be able to prove, or disprove, the previous conjecture instantaneously. 
However the previous conjecture, provided it holds true, together with the first two statements of Theorem 
\ref{main} will only settle the case $a=3, b=2$ of Problem \ref{probone}. In case of generic $a$ and 
$b$ the problem seems elusive and, as a first step, one may consider the problem of obtaining 
upper bounds on the Lebesgue measure of $T(a,b)$-extremal sets.  
The approach of Theorem \ref{polar} could work in the general case, provided one can 
solve the following geometric problem that is interesting on its own. 
Our paper ends with a notion of generalised diameter of a set and the formulation of this problem. 

Fix positive integers, $a,b$ such that $a>b\geq 2$ and 
a subset, $S$, of $\mathbb{R}^d, d\geq 2$. If $\{p_1,\ldots,p_b\} \subseteq S$, is a subset of cardinality $b$, set
\[ \Delta(\{p_1,\ldots,p_b\}) = \max_{1\leq i<j\leq b} |p_i-p_j| \]
to be the maximum distance between the points $\{p_1,\ldots,p_b\}$.
Given a finite set of points, $F\subseteq S$, let $|F|$ denote its cardinality. If $|F|>b$ let 
\[ \binom{F}{b} := \{F'\subseteq F : |F'| =b \} \]
be the class consisting of all subsets of $F$ of cardinality $b$.
Finally, define the $(a,b)$-\emph{diameter} of $S$ by 
\[ \text{diam}_{a,b}(S) = \sup_{\{F\subseteq S: |F|=a\}}\;  \min_{F'\in\binom{F}{b}} \Delta(F') . \]

\begin{prbl} Fix $\delta>0$ and let $S$ be a subset of $\mathbb{R}^d,d\geq 2$, 
such that $\text{diam}(S)=\delta$. Assume that 
there exist pairs $\{(a_i,b_i)\}_{i=1}^{m}$ with $a_i>b_i\geq 2$
and positive real numbers $r_1,\ldots,r_m$ such that 
\[ \text{diam}_{a_i,b_i}(S) = r_i, \; \text{for} \; i=1,\ldots,m . \]
What is the smallest radius of a disk that contains $S$?
\end{prbl}

\section*{Acknowledgements} 
This work is supported by ERC Starting 
Grant 240186 "MiGraNT, Mining Graphs and Networks: a Theory-based approach". I am grateful to 
Themis Mitsis, Tobias M\"uller, Konstantinos Pelekis and 
Nikos Pelekis for many valuable discussions, comments and suggestions.

\end{document}